\documentclass{article}

\usepackage{times}
\usepackage{graphicx} % more modern
\usepackage{subfigure}
\newcommand{\indep}{\rotatebox[origin=c]{90}{$\models$}}

\usepackage{natbib}

\usepackage{amsmath}
\usepackage{amsthm}
\usepackage{amsfonts}
\usepackage{algorithm}
\usepackage{algorithmic}

\newtheorem{lemma}{Lemma}

\usepackage[accepted]{icml2014}

\icmltitlerunning{Understanding Protein Dynamics with $L_1$-Regularized Reversible Hidden Markov Models}

\begin{document}

\twocolumn[
\icmltitle{Understanding Protein Dynamics with\\ $L_1$-Regularized Reversible Hidden Markov Models}

% unless you've provided the [accepted] option to the icml2014
\icmlauthor{Robert T. McGibbon}{rmcgibbo@stanford.edu}
\icmladdress{Department of Chemistry, Stanford University, Stanford CA 94305, USA}
\icmlauthor{Bharath Ramsundar}{rbharath@stanford.edu}
\icmladdress{Department of Computer Science, Stanford University, Stanford CA 94305, USA}
\icmlauthor{Mohammad M. Sultan}{msultan@stanford.edu}
\icmladdress{Department of Chemistry, Stanford University, Stanford CA 94305, USA}
\icmlauthor{Gert Kiss}{gkiss@stanford.edu}
\icmladdress{Department of Chemistry, Stanford University, Stanford CA 94305, USA}
\icmlauthor{Vijay S. Pande}{pande@stanford.edu}
\icmladdress{Department of Chemistry, Stanford University, Stanford CA 94305, USA}

\icmlkeywords{protein, dynamics, biological, hmm, markov, lasso}

\vskip 0.3in
]

\begin{abstract}
\label{abstract}
We present a machine learning framework for modeling protein dynamics. Our
approach uses $L_1$-regularized, reversible hidden Markov models to
understand large protein datasets generated via molecular dynamics
simulations. Our model is motivated by three design principles: (1) the
requirement of massive scalability; (2) the need to adhere to relevant
physical law; and (3) the necessity of providing accessible
interpretations, critical for both cellular biology and rational drug
design. We present an EM algorithm for learning and introduce a model
selection criteria based on the physical notion of convergence in
relaxation timescales. We contrast our model with standard methods in
biophysics and demonstrate improved robustness. We implement our algorithm
on GPUs and apply the method to two large protein simulation datasets
generated respectively on the NCSA Bluewaters supercomputer and the
Folding@Home distributed computing network. Our analysis identifies the
conformational dynamics of the ubiquitin protein critical to cellular
signaling, and elucidates the stepwise activation mechanism of the c-Src
kinase protein.
\end{abstract}

\section{Introduction}
\label{introduction}
Protein folding and conformational change are grand challenge problems, relevant to a multitude of human diseases, including Alzheimer's disease, Huntington's disease and cancer. These problems entail the characterization of the process and pathways by which proteins fold to their energetically optimal configuration and the dynamics between multiple long-lived, or ``metastable,'' configurations on the potential energy surface. Proteins are biology's molecular machines; a solution to the folding and conformational change problem would deepen our understanding of the mechanism by which microscopic information in the genome is manifested in the macroscopic phenotype of organisms. Furthermore, an understanding of the structure and dynamics of proteins is increasingly important for the rational design of targeted drugs \cite{doi:10.1146/annurev.pharmtox.43.100901.140216}.

Molecular dynamics (MD) simulations provide a computational microscope by which protein dynamics can be studied with atomic resolution  \cite{dill1995principles}. These simulations entail the forward integration of Newton's equations of motion on a classical potential energy surface. The potential energy functions in use, called forcefields, are semi-emprical approximations to the true quantum mechanical Born-Oppenheimer surface, designed to reproduce experimental observables \cite{Beauchamp2012Protein}. For moderately sized proteins, this computation can involve the propagation of more than a million physical degrees of freedom.  Furthermore, while folding events can take milliseconds ($10^{-3}$ s) or longer, the simulations must be integrated with femtosecond ($10^{-15}$ s) timesteps, requiring the collection of datasets containing trillions of data points.

While the computational burden of performing MD simulations has been a central challenge in the field, significant progress has been achieved recently with the development of three independent technologies: ANTON, a special-purpose supercomputer using a custom ASIC to accelerate MD \cite{Shaw:2007:ASM:1250662.1250664}; Folding@Home, a distributed computing network harnessing the desktop computers of more than 240,000 volunteers \cite{Shirts08122000}; and Google Exacycle, an initiative utilizing the spare cycles on Google's production infrastructure for science \cite{kohlhoff2014cloud}.

The analysis of these massive simulation datasets now represents a major difficulty: how do we turn data into knowledge \cite{lane2012milliseconds}? In contrast to some other machine learning problems, the central goal here is not merely prediction. Instead, we view analysis -- often in the form of probabilistic models generated from MD datasets -- as a tool for generating scientific insight about protein dynamics.

Useful probabilistic models must embody the appropriate physics. The guiding physical paradigm by which chemical dynamics are understood is one of \emph{states} and \emph{rates}. States correspond to metastable regions in the configuration space of the protein and can often be visualized as wells on the potential energy surface. Fluctuations within each metastable state are rapid;  the dominant, long time-scale dynamics can be understood as a jump process moving with various rates between the states. This paradigm motivates probabilistic models based on a discrete-state Markov chain. \emph{A priori}, the location of the metastable states are unknown. As a result, each metastable state should correspond to a latent variable in the model. Hidden Markov models (HMMs) thus provide the natural framework.

Classical mechanics at thermal equilibrium satisfy a symmetry with respect to time: a microscopic process and its time-reversed version obey the same laws of motion. The stochastic analogue of this property is reversibility (also called detailed balance): the equilibrium flux between any two states $X$ and $Y$ is equal in both directions. Probabilistic models which fail to capture this essential property will assign positive probability to systems that violate the second law of thermodynamics \cite{Prinz2011Markov}. Hence, we enforce detailed balance in our HMMs.

In addition to the constraints motivated by adherence to physical laws, suitable probabilistic models should, in broad strokes, incorporate knowledge from prior experimental and theoretical studies of proteins. Numerous studies indicate that only a subset of the degrees of freedom are essential for describing the protein's dominant long time-scale dynamics (see \citet{Cho17012006} and references therein). Furthermore, substantial prior work indicates that protein folding occurs via a sequence of localized shifts \cite{Maity29032005}. Together, these pieces of evidence motivate the imposition of L$_1$-fusion regularization \cite{tibshirani2005sparsity}. The L$_1$ term penalizes deviations amongst states along uninformative degrees of freedom, thereby suppressing their effect on the model. Furthermore, the pairwise structure of the fusion penalty minimizes the number of transitions which involve global changes: many pairs of states will only differ along a reduced subset of the dimensions.

The main results of this paper are the formulation of the L$_1$-regularized reversible HMM and the introduction of a simple and scalable learning algorithm to fit the model. We contrast our approach against standard frameworks for the analysis of MD data and demonstrate improved robustness and physical interpretability.

This paper is organized as follows. Section 2 describes prior work. Section 3 introduces the model and associated learning algorithm. Section 4 applies the model to three systems: a toy double well potential; ubiquitin, a human signaling protein; and c-Src kinase, a critical regulatory protein involved in cancer genesis. Section 5 provides discussion and indicates future directions.

\section{Prior Work}
\label{priorwork}
Earlier studies have applied machine learning techniques to investigate protein structure prediction -- the problem of discovering a protein's energetically optimal configuration -- using CRFs, belief propagation, deep learning, and other general ML methods \cite{sontag2012tightening, di2012deep, chu:2006:BSM:1137243.1137470, baldi:2003:PDL:945365.945379}. But proteins are fundamentally dynamic systems, and none of these approaches offer insight into kinetics; rather, they are concerned with extracting static information about protein structure.

The dominant computational tool for studying protein dynamics is MD. Traditional analyses of MD datasets are primarily visual and non-quantitative. Standard approaches include watching movies of a protein's structural dynamics along simulation trajectories, and inspecting the time evolution of a small number of pre-specified degrees of freedom \cite{Humphrey199633, Karplus10052005}. While these methods have been successfully applied to smaller proteins, they struggle to characterize the dynamics of the large and complex biomolecules critical to biological function. Quantitative methods like PCA can elucidate important (high variance) degrees of freedom, but fail to capture the rich temporal structure in MD datasets.

Markov state models (MSMs) are a simple class of probabilistic models, recently introduced to capture the temporal dynamics of the folding process. In an MSM, protein dynamics are modeled by the evolution of a Markov chain on a discrete state space. The finite set of states is generated by clustering the set of configurations in the MD trajectories \cite{doi:10.1021/ct200463m}. MSMs can be viewed as fully observable HMMs. More recently, HMMs with multinomial emission distributions have been employed on this discrete state space \cite{:/content/aip/journal/jcp/139/18/10.1063/1.4828816}.

Although MSMs have had a number of notable successes \cite{doi:10.1021/ja302528z, Sadiq26112012}, they are brittle and complex.  Traditional MSMs lack complete data likelihood functions, and learning cannot be easily characterized by a single optimization problem. For these reasons, MSM learning requires significant manual tuning. For example, because clustering is purely a preprocessing step, the likelihood function contains no guidance on the choice of the metastable states. Moreover, the lack of uncertainty in the observation model necessitates the introduction of a very large number of states, typically more than ten thousand, in order to cover the protein's phase space at sufficient resolution. This abundance of states is statistically inefficient, as millions of pairwise transition parameters must be estimated in typically-sized models, and renders interpretation of learned MSMs challenging.

\section{Fusion $L_1$-Regularized Reversible HMM}
\label{model}
We introduce the $L_1$-regularized reversible HMM with Gaussian emissions, a generative probabilistic model over multivariate discrete-time continuous-space time series. As discussed in Section 1, we integrate necessary physical constraints on top of the core hidden Markov model \cite{rabiner1986introduction}.

Let $\{Y_t\}$ be the observed time series in $\mathbb{R}^D$ of length $T$ (\textit{i.e.}, the input simulation data), and let $\{X_t\}$ be the
corresponding latent time series in $\{1, \ldots, K\}$, where $K$ is a hyperparameter indicating the number of hidden states in the model. Each
hidden variable $x_t$ corresponds to a metastable state of the physical system. The emission distribution given $X_t=k$ is a multivariate normal distribution parameterized by mean $\mu_k \in \mathbb{R}^D$ and diagonal covariance matrix $\operatorname{Diag}(\sigma^2_k) \in \mathbb{R}^{D\cdot D}$ (where $\sigma^2_k \in \mathbb{R}^D$ is the vector of diagonal covariance elements). We use the notation $(\mu_{k})_{j}$ to indicate the $j$th element of the vector $\mu_{k}$.

Controlling the means $\{ \mu_k \}$ is critical for achieving physically interpretable models. As discussed in Section 1, we wish to minimize the differences between $\mu_k$ and $\mu_{k'}$ to the extent possible. Consequently, we place a fusion $L_1$ penalty \cite{tibshirani2005sparsity} on our log likelihood function, which adds the following pairwise cost:
\begin{align*}
\lambda \sum_{k,k'} \sum_{j} \tau^{(j)}_{k,k'} \left| (\mu_{k})_{j} -
(\mu_{k'})_{j} \right|.
\end{align*}
Here, $\lambda$ governs the overall strength of the penalty, while the adaptive fusion weights, $\{ \tau^{(j)}_{k,k'} \}$, control the contribution from each pair of states \cite{guo2010pairwise}. During learning, the adaptive fusion weights are computed as
\begin{align*}
\tau^{(j)}_{k,k'} = | (\tilde{\mu}_{k})_{j} - (\tilde{\mu}_{k'})_{j} |^{-1},
\end{align*}
where the $\{\tilde{\mu}_{k}\}$ are the learned metastable state means in the absence of the penalty. The intuition motivating the adaptive strength of the penalty is that if degree of freedom $j$ is informative for separating states $k$ and $k'$, the corresponding fusion penalty should be applied lightly.

The reversible time evolution of the model is parameterized by an irreducible, aperiodic, row-normalized $K$ by $K$ stochastic matrix $\mathbf{T}$, which satisfies detailed balance. Mathematically, the detailed balance constraint is
\begin{align*}
\forall k, k',\ \pi_k \mathbf{T}_{k, k'} &= \pi_{k'}\mathbf{T}_{k',k},
\end{align*}
where row vector $\pi$ is the stationary distribution of $\mathbf{T}$. The stationary distribution $\pi$ also parameterizes the initial distribution over the metastable states. By the Perron--Frobenius theorem, $\pi$ is the dominant left eigenvector of $\mathbf{T}$ with eigenvalue 1 and is not an independent parameter in this model.

The initial distributions and evolution of $\{X_t, Y_t\}$ satisfy the following equations:
\begin{align*}
    X_0 &\sim \sum_{k=1}^{K}\pi_k \delta_k, \\
    X_{t+1} &\sim \sum_{k=1}^{K} \mathbf{T}_{X_t, k} \, \delta_k, \\
    Y_{t} &\sim \mathcal{N}(\mu_{X_t}, \sigma^2_{X_t}).
\end{align*}
The complete data likelihood $\{x_t, y_t\}$ is
\begin{align*}
&\mathcal{L}(\{x_t\}, \{y_t\} | \mathbf{T}, \mu, \sigma) = \\
&\qquad \pi_{x_0} \prod_{t=1}^{T-1} \mathbf{T}_{x_{t-1}, x_t} \prod_{t=0}^{T-1} \mathcal{N}(y_t; \mu_{x_t}, \sigma^2_{x_t}).
\end{align*}

The hyperparameter $\Delta$ controls the discretization interval at which a protein's coordinates are sampled to obtain $\{y_t\}$. In the absence of downsampling by $\Delta$, subsequent samples $y_t, y_{t+1}$ would be highly correlated. On the other hand, subsequent samples from an HMM are conditionally independent given the hidden state. Choice of $\Delta$ large enough recovers this conditional independence (\emph{vide infra}).

\subsection{Learning}
The model is fit using expectation-maximization. The E-step is standard, while the M-step requires modification to enforce the detailed balance constraint on $\mathbf{T}$ and the adaptive fusion penalty on the $\{ \mu_k\}$.

\subsubsection{E-step}
Inference is identical to that for the standard HMM, using the forward-backward algorithm \cite{rabiner1986introduction} to compute the following quantities:
\begin{align*}
\gamma_i(t) &= \mathbb{P}(X_t = i| \{ y_t \}),\\
\xi_{ij}(t) &= \mathbb{P}(X_t = i, X_{t+1} = j\mid \{ y_t \}).
\end{align*}

\subsubsection{M-step}
Both the penalty on $\{\mu_k\}$ and the reversibility constraint affect only the M-step. The M-step update to the means in the $t$-th iteration of EM consists of maximizing the penalized log-likelihood function
\begin{align*}
\mu_k^{(t+1)} = \underset{\mu_k}{\operatorname{argmin}} & \sum_i^N
\sum_k^K \gamma_k(i) \frac{(x_i - \mu_k)^2}{2(\sigma_k^2)^{(t)}} \\
 &+ \lambda \sum_{k, k'} \sum_{j} \tau_{k,k'}^{(j)} \left| \mu_{k,j} -
 \mu_{k',j} \right|.
\end{align*}
The $\{\mu_k\}$ update is a quadratic program, which can be solved by a variety of methods. We compute $\{\mu_k\}$ by iterated ridge regression. Following \citet{guo2010pairwise} and \citet{fan2001variable}, we use the local quadratic approximation
\begin{align*}
& \left|\ \mu_{k,j}^{(t, s+1)} - \mu_{k',j}^{(t, s+1)} \right| \approx \\
&\hspace{5em} \frac{\left(\mu_{k,j}^{(t, s+1)} - \mu_{k',j}^{(t, s+1)}\right)^2}{2 \, \left|\,
\mu_{k,j}^{(t, s)} - \mu_{k',j}^{(t, s)} \,\right|} +  \frac{1}{2} \left|\,
\mu_{k,j}^{(t, s)} - \mu_{k',j}^{(t, s)} \right|.
\end{align*}
where $s$ is the iteration index for this procedure within the $t$-th M-step. This approximation is based on the identity
\begin{align*}
|x-y| = \frac{(x-y)^2}{2\,|x-y|} + \frac{1}{2} |x-y|.
\end{align*}
Under the approximation, we obtain a generalized ridge regression problem which can be solved in closed form during each iteration $s$. Note that this approximation is only valid when $| \mu^{(t, s)}_{k,j} - \mu^{(t, s)}_{k',j} | > 0$. For numerical stability, we threshold $| \mu^{(t, s)}_{k,j} - \mu^{(t, s)}_{k',j} | $ to zero at values less than $10^{-10}$.

The variance update is standard:
\begin{align*}
\sigma_k^2 = \frac{\sum_t \gamma_k(t) (y_t - \mu_k)^T(y_t - \mu_k)}{\sum_t \gamma_k(t)}.
\end{align*}
The transition matrix update is
\begin{align*}
\mathbf{T} = \arg\max_{\mathbf{T}} \sum_{ij} \log(\mathbf{T}_{ij}) \sum_t
\xi_{ij}(t).
\end{align*}
Because the Gaussian emission distributions have infinite support, $\mathbf{T}$ is irreducible and aperiodic by construction. However, we must explicitly constrain $\mathbf{T}$ to satisfy detailed balance.
\begin{lemma}
    $\mathbf{T}$ satisfies detailed balance if and only if $\mathbf{T}_{ij} = \frac{\mathbf{W}_{ij}}{\sum_k \mathbf{W}_{ik}}$, where $\mathbf{W} = \mathbf{W}^T$.
\end{lemma}
\begin{proof}
If $\mathbf{T}$ satisfies detailed balance, then let $\mathbf{W}_{ij} = \pi_i \mathbf{T}_{ij} =  \pi_{j}\mathbf{T}_{ji} = \mathbf{W}_{ji}$. Then note
\begin{align*}
\frac{\mathbf{W}_{ij}}{\sum_k \mathbf{W}_{ik}} = \frac{\pi_i \mathbf{T}_{ij}}{\sum_k \pi_i \mathbf{T}_{ik}}
= \frac{\mathbf{T}_{ij}}{\sum_k \mathbf{T}_{ik}} = \mathbf{T}_{ij}
\end{align*}
To prove the converse, assume $\mathbf{T}_{ij} = \frac{\mathbf{W}_{ij}}{\sum_k \mathbf{W}_{ik}}$, with $\mathbf{W} = \mathbf{W}^T$. Let $\pi_i = \sum_k \mathbf{W}_{ik}$. Then $\pi_i \mathbf{T}_{ij} = \mathbf{W}_{ij} = \mathbf{W}_{ji} = \pi_j \mathbf{T}_{ji}$.
\end{proof}

Substituting the results of Lemma 1, we rewrite the transition matrix update as
\begin{align*}
\mathbf{W} = \arg\max_{\mathbf{W}} \Big(\Big[ \sum_{ij} \log(\mathbf{W}_{ij}) - \log \pi_i \Big] \sum_t \xi_{ij}(t) \Big).
\end{align*}
We compute the derivative of the inner term with respect to $\log \mathbf{W}_{ij}$ and optimize with L-BFGS \cite{nocedal2006numerical}.

\subsection{Model Selection}
\label{subsec:model_selection}
There are two free model parameters: $K$ and $\Delta$. The number of metastable states, $K$, is expected to be small -- at most a few dozen. To choose $K$, we can use the AIC or BIC selection criteria, or alternatively enumerate a few small values.

The choice of $\Delta$ is more difficult than the choice of $K$, as changing the discretization interval alters the support of the likelihood function. Recall that choosing $\Delta$ too small results in subsequent samples $y_t, y_{t+1}$ becoming highly correlated, while the model satisfies the conditional independence assumption $Y_{t} \,\indep\, Y_{t+1} \,|\, X_t$. Moreover, small $\Delta$ increases data-storage requirements, while $\Delta$ too large will needlessly discard data. Thus a balance between these two conflicting directives is necessary.

We use the physical criterion of convergence in the relaxation timescales to evaluate when $\Delta$ is large enough. The propagation of the dynamics from an initial distribution over the hidden states, $X_t$, can be described by
\begin{align*}
P(y_{t+n} \,|\, x_t) = \sum_{k=0}^{K-1} \mathcal{N}(y_{t+n}; \mu_k,
\sigma^2_k) \, \left(x_t^T \, \mathbf{T}^n \right)_k. 
\end{align*}
Diagonalize $\mathbf{T}$ in terms of its left eigenvectors $\phi_i$, right eigenvectors $\psi_i$, and eigenvalues $\lambda_i$ (such a diagonalization is always possible since $\mathbf{T}$ is a stochastic matrix).
\begin{align*}
P(y_{t+n} \,|\, x_t) = \sum_{k=1}^{K} \left[ \mathcal{N}(y_{t+n}; \mu_k, \sigma^2_k) \left( \sum_{i=1}^{K} \lambda_i^n \, \langle x_t, \psi_i \rangle \, \phi_i \right)_k \right]
\end{align*}
Since $\pi$ is the stationary left eigenvector of $\mathbf{T}$, and the remaining eigenvalues lie in the interval $-1 < \lambda_i < 1$, the collective dynamics can be interpreted as a sum of exponential relaxation processes.
\begin{align*}
P(y_{t+n} \,|\, x_t) = \sum_{k=1}^{K} \left[ f_k(y_{t+n}) \! \left( \pi + \sum_{i=2}^{K} e^{-n/\tau_i} \, \langle x_t, \psi_i \rangle \, \phi_i \right)_k \right]
\end{align*}
In the equation, we define $f_k(y) = \mathcal{N}(y; \mu_k, \sigma^2_k)$. Each eigenvector of $\mathbf{T}$ (except the first) describes a dynamical mode with characteristic relaxation timescale
\begin{align*}
\tau_i=-\frac{1}{\ln \lambda_i}.
\end{align*}
The longest timescales, $\tau_i$, are of central interest from a molecular modeling perspective because they describe dynamical modes visible in time-resolved protein experiments \cite{doi:10.1021/jp109592b} and are robust against perturbations \cite{weber2012protein}. We choose $\Delta$ large enough to converge the $\tau_i$: for adequately large $\Delta$, we expect $\tau_i(\Delta)$ to asymptotically converge to the true relaxation timescale $\tau_i^*$. For simple systems, we may evaluate $\tau_i^*$ explicitly, while for larger systems, we choose $\Delta$ large enough so that $\tau_i(\Delta)$ no longer changes with further increase in the discretization interval. 

\subsection{Implementation}
We implement learning for both multithreaded CPU and NVIDIA GPU platforms. In the CPU implementation, we parallelize across trajectories during the E-step using OpenMP. The largest portion of the run time is spent in {\sc log-sum-exp} operations, which we manually vectorize with SSE2 intrinsics for SIMD architectures. Parallelism on the GPU is more fine grained. The E-step populates two $T \times K \times K$ arrays with forward and backwards sweeps respectively. To fully utilize the GPU's massive parallelism, each trajectory has a team of threads which cooperate on updating the $K \times K$ matrix at each time step. Specialized CUDA kernels were written for $K=4, 8, 16$ and 32 along with a generic kernel for $K>32$.

Even in log space, for long trajectories, the forward-backward algorithm can suffer from an accumulation of floating point errors which lead to catastrophic cancelation during the computation of $\gamma_i(t)$. This risk requires that the forward-backward matrices be accumulated in double precision, whereas the rest of the calculation is safe in single precision.

The speedup using our GPU implementation is $15\times$ compared to our optimized CPU implementation and $75\times$ with respect to a standard numpy implementation using $K=16$ states on a NVIDIA GTX TITAN GPU / Intel Core i7 4 core Sandy Bridge CPU platform. Further scaling of the implementation could be achieved by splitting the computation over multiple GPUs with MPI.

\section{Experiments}
\label{experiments}
\begin{figure}
\includegraphics[width=3.4in]{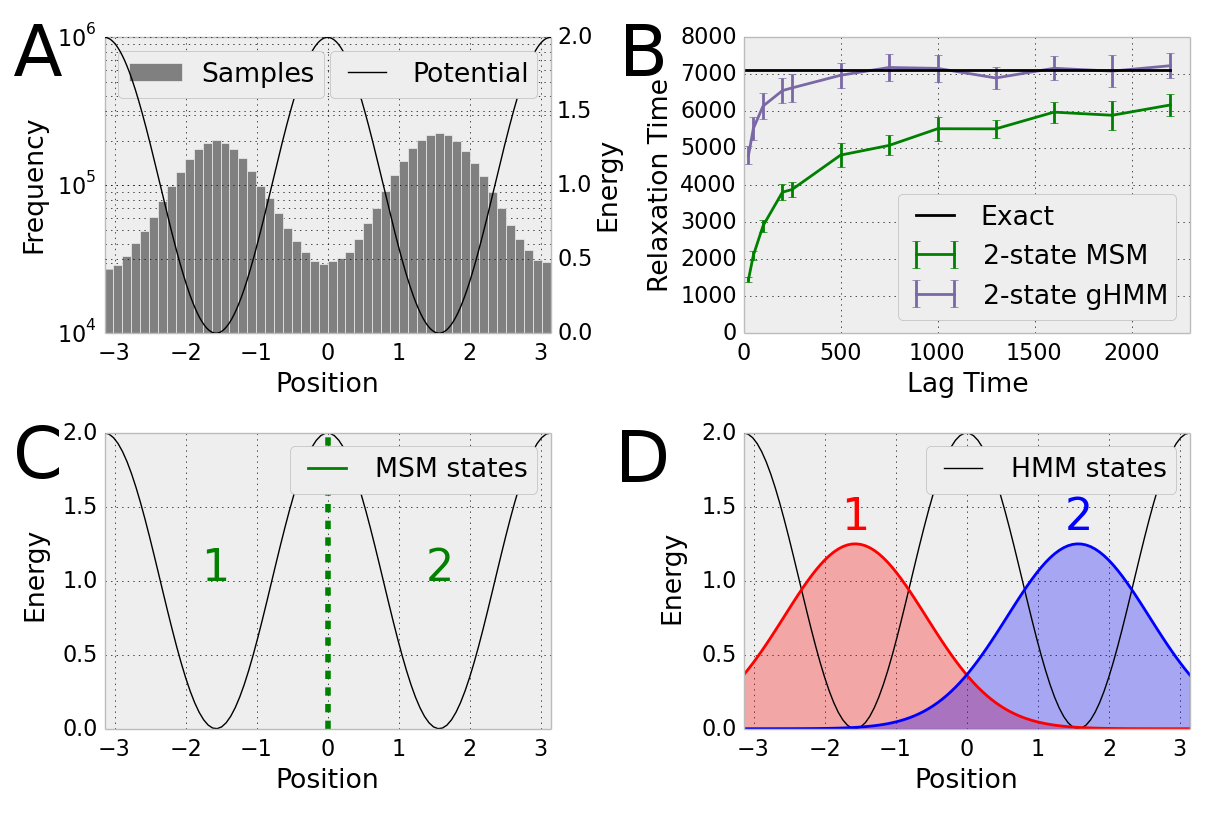}
\caption{Simulations of Brownian dynamics on a double well potential ({\bf A}) illustrate the advantages of the HMM over the MSM. When the dynamics are discretized at a time interval of $>500$ steps, the 2-state HMM, unlike the 2-state MSM achieves a quantitatively accurate prediction of the first relaxation timescale ({\bf B}). The MSM ({\bf C}) features hard cutoffs between the states wheres the HMM ({\bf D}) each have infinite support.}
\end{figure}

\subsection{Double Well Potential}
We first consider a one-dimensional diffusion process $y_t$ governed by Brownian dynamics. The process is described by the stochastic differential equation
\begin{align*}
\frac{dy_t}{dt} = -\nabla V(y_t) + \sqrt{2D} R(t)
\end{align*}
where V is the reduced potential energy, $D$ is the diffusion constant, and $R(t)$ is a zero-mean delta-correlated stationary Gaussian process. For simplicity, we set $D=1$ and consider the double well potential
\begin{align*}
V(y) = 1 + \cos (2y)
\end{align*}
with reflecting boundary conditions at $y=-\pi$ and $y=\pi$. Using the Euler-Maruyama method and a time step of $\Delta t=10^{-3}$, we produced ten simulation trajectories of length $5 \times 10^5$ steps each. The histogrammed trajectories are shown in Fig. 1({\bf A}). The exact value of the first relaxation timescale was computed by a finite element discretization of the corresponding Fokker-Planck equation \cite{higham2001algorithmic}.

We applied both a two-state MSM and two-state HMM, with fusion $L_1$ regularization parameter $\lambda=0$, to the simulation trajectories. The MSM states were fixed, with a dividing surface at $y=0$, as shown in Fig. 1({\bf C}). The HMM states were learned, as shown in Fig. 1({\bf D}). Both the MSM and the HMM display some sensitivity with respect to the discretization interval, with more accurate predictions of the relaxation timescale at longer lag times. 

The two-state MSM is unable to accurately learn the longest timescale, $\tau_1$, even with large lag times, while the two-state HMM succeeds in identifying $\tau_1$ with $\Delta \geq 500$ Fig. 1({\bf B}).

\subsection{Ubiquitin}
\begin{figure}[h]
\centering
\includegraphics[width=2.25in]{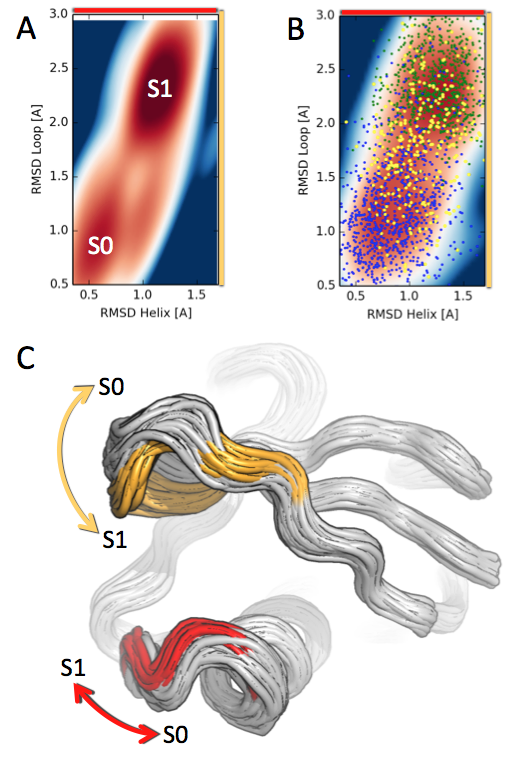}
\caption{Dynamics of Ub. {\bf(A)} The HMM identifies the two metastable states of Ub, varying primarily in the loop and helix regions (axes in yellow and red respectively). {\bf(B)} The MSM fails to cleanly separate the two underlying physical states. Three post-processed macrostates from the MSM are shown (in blue, green, and yellow). {\bf(C)} A structural rendering of the conformational states of the Ub system. S0, shown in grey, binds to the UCH family of proteins, and S1 (with characteristic structural differences to S0 in red and yellow) binds to the USP family.}
\label{fig:Ub}
\end{figure}
Ubiquitin (Ub) is a regulatory hub protein at the intersection of many signaling pathways in the human body \cite{hershko1998ubiquitin}. Among its many tasks are the regulation of inflammation, repair of DNA, and the breakdown and recycling of waste proteins. Ubiquitin interacts with close to 5000 human signaling proteins. Understanding the link between structure and function in ubiquitin would elucidate the underlying framework of the human signaling network.  

We obtained a dataset of MD simulations of human Ub consisting of 3.5 million data points. The protein, shown in Fig. \ref{fig:Ub}, is composed of 75 amino acids. The simulations were performed on the NCSA Blue Waters supercomputer. The resulting structures were featurized by extracting the distance from each amino acid's central carbon atom to its position in the simulations' starting configurations. HMMs were constructed with 2 to 6 states. We chose $\Delta$ by monitoring the convergence of the relaxation timescales as discussed in Sec.  \ref{subsec:model_selection}, and set the $L_1$ fusion penalty heuristically to a default value of $\lambda=0.01$.
In agreement with existing biophysical data \cite{zhang2012conformational}, the HMMs correctly determined that Ub was best modeled with 2 states (Fig. \ref{fig:Ub}{\bf A}). For ease of representation, the learned HMM is shown projected onto two critical degrees of freedom (discussed below).

For comparison, we generated MSM models with 500 microstates (Fig. \ref{fig:Ub}{\bf B}) and projected upon the same critical degrees of freedom. We used a standard kinetic lumping post-processing step to identify 3 macrostates (shown in green, blue, and yellow respectively); the lumping algorithm collapsed when asked to identify 2 macrostates \cite{Bowman2012Improved}. Contrast the simple, clean output of the 2 state HMM in Fig. \ref{fig:Ub}({\bf A}) with the standard MSM of Fig. \ref{fig:Ub}({\bf B}). Note how significant post-processing and manual tuning would be required to piece together the true two-state structural dynamics of Ub from the MSM output.

We display a structural rendering of the Ub system in Fig. \ref{fig:Ub}({\bf C}). The imposed $L_1$ penalty of the HMM suppresses differences among the uninformative degrees of freedom depicted in grey. The remaining portions of the protein (shown in color) reveal the two critical axes of motion of the Ub system: the hinge dynamics of the loop region displayed in yellow and a kink in the lower helix shown in red. We use these axes in the simplified representations shown in Figs. \ref{fig:Ub}({\bf A,B}). 

The states S0 and S1 identified by the HMM have direct biological interpretations. Comparison to earlier experimental work reveals that configuration S0 binds to the UCH family of proteins, while configuration S1 binds to the USP family instead \cite{komander2009breaking}. The families play differing roles in the vital task of regenerating active Ub for the cell-signaling cycle.

Together, MD and the HMM analysis provide atomic insight into the effect of protein structure on ubiquitin's role in the signaling network. Our analysis approach may have significant value for protein biology and for the further study  of cellular signaling networks.  Although experimental studies of protein signaling provide the gold standard for hard data, they struggle to provide structural explanations --- knowing why a certain protein is more suited for certain signaling functions is challenging at best. In contrast, the MD/HMM approach can provide a direct link between structure and function and give a causal basis for observed protein activity.

\subsection{c-Src Tyrosine Kinase}

Protein kinases are a family of enzymes that are critical for regulating cellular growth whose aberrant activation can lead to uncontrolled cellular proliferation. Because of their central role in cell proliferation, kinases are a critical target for anti-cancer therapeutics. The c-Src tyrosine kinase is a prominent member of this family that has been implicated in numerous human malignancies \cite{pmid20689754}.

Due to the protein's size and complexity, performing MD simulations of the c-Src kinase is a formidable task. The protein, shown in Fig.  \ref{fig:kinase}{\bf A}, consists of 262 amino acids; when surrounding water molecules -- necessary for accurate simulation -- are taken into account, the system has over 40,000 atoms. Furthermore, transition between the active and inactive states takes hundred of microseconds. Adequate sampling of these processes therefore requires hundreds of billions of MD integrator steps. Simulations of the c-Src kinase were performed on the Folding@Home distributed computing network, collecting a dataset of 4.7 million configurations from 550 $\mu s$ of sampling, for a total of 108 GB of data \cite{Shukla2014Activation}. 

In order to understand the molecular activation mechanism of the c-Src kinase, we analyzed this dataset using the $L_1$ regularized reversible HMM. We featurized the configurations by extracting the distance from each amino acid's central carbon atom to its position in an experimentally determined inactive configuration. We built HMMs with 2 to 6 states, and singled out the 3 state model for achieving a balance of complexity and interpretability. As with Ub, we chose $\Delta$ by monitoring the convergence of the relaxation timescales, same default $L_1$ fusion penalty of $\lambda=0.01$.

The $L_1$-regularized reversible HMM elucidates the c-Src kinase activation pathway, revealing a stepwise mechanism of the dynamics. A projection of the learned HMM states onto two key degrees of freedom is shown in Fig. \ref{fig:kinase}{\bf B}. Fig. \ref{fig:kinase}{\bf C} shows a structural representation of the means of the three states, highlighting a sequential activation mechanism. The transformation from the inactive to the intermediate state occurs first by the unfolding of the A-loop (the subsection of the protein highlighted in red). Activation is completed by the inward rotation of the C-helix (highlighted in orange) and rupture of a critical side chain interaction between two amino acids on the C-helix and the A-loop respectively.

Although the protein structure is complex, the activation process takes place only in a small portion of the overall protein; the random fluctuations of the remaining degrees of freedom are largely uncoupled from the activation process. As with Ub, the $L_1$ penalty suppresses the signal from unimportant degrees of freedom shown in grey. In contrast to the simplicity of HMM approach, a recent MSM analysis of this dataset found similar results, but required 2,000 microstastates and significant post-processing of the models to generate physical insight into the activation mechanism \cite{Shukla2014Activation}.

The identification of the intermediate state along the activation pathway has substantial implications in the field of rational drug design. Chemotherapy drugs often have harmful side effects because they target portions of proteins that are common across entire families, interfering with both the uncontrolled behavior of tumor proteins as well as the critical cellular function of healthy proteins. Intermediate states, such as the one identified by the HMM, are more likely to be unique to each kinase protein; future therapeutics that target these intermediate states could have significantly fewer deleterious side effects \cite{doi:10.1021/cb300663j}.

\begin{figure}
\centering
\includegraphics[width=3.2in]{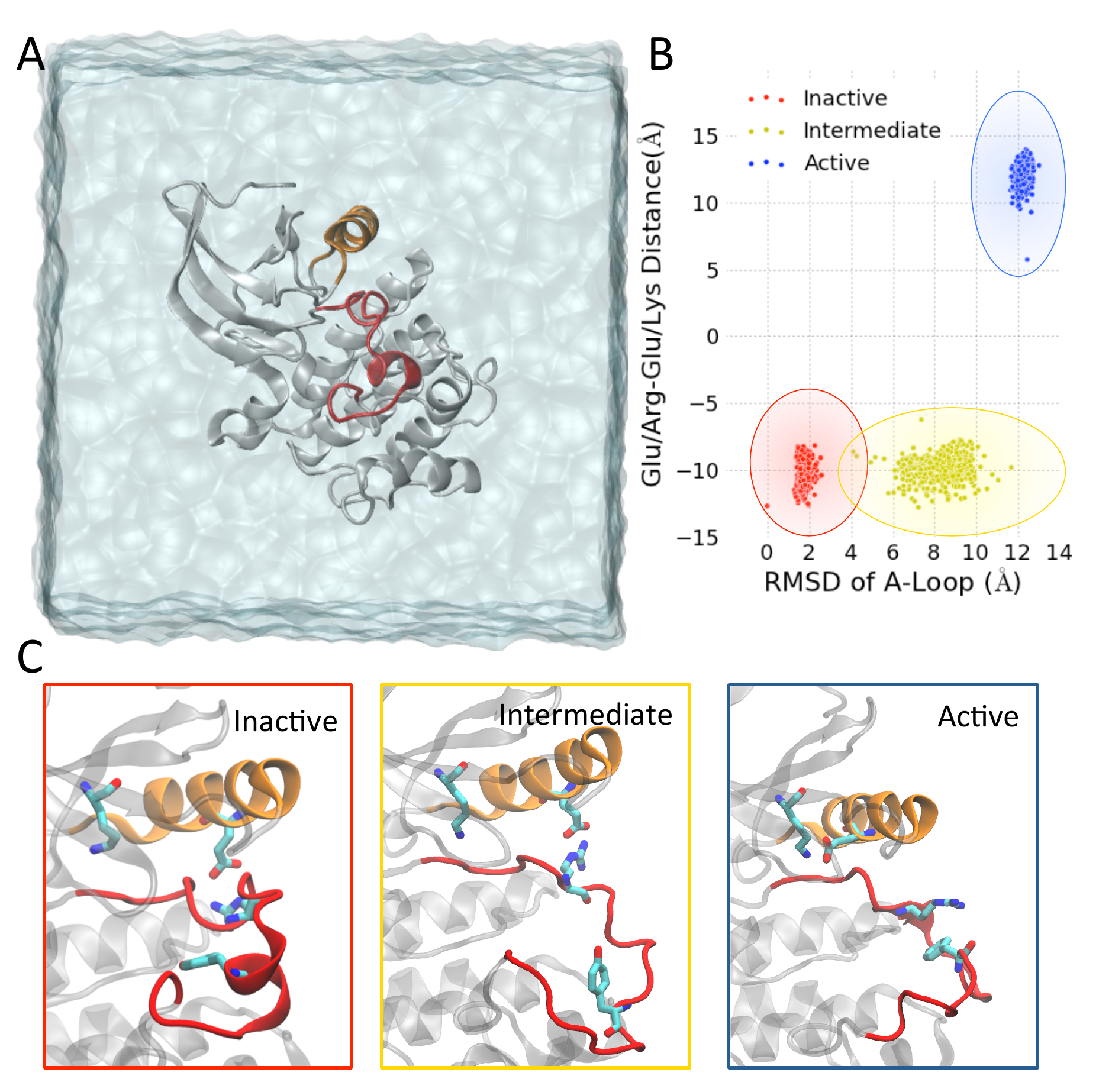}
\caption{Activation of the c-Src Kinase. ({\bf A}) Structure of the protein system. ({\bf B}) The 3 state HMM, projected onto two degrees of freedom representing the positions of the A-loop (shown in red) and C-helix (shown in orange) respectively. ({\bf C}) Structural renderings of the means of the hidden states showing atomistic details of the activation pathway.}
\label{fig:kinase}
\end{figure}

\section{Discussion and Conclusion}
\label{discussion}
Currently, MSMs are a dominant framework for analyzing protein dynamics datasets. We propose replacing this methodology with $L_1$-regularized reversible HMMs. We show that HMMs have significant advantages over MSMs: whereas the MSM state decomposition is a prepreprocessing procedure without guidance from a complete-data likelihood function, the HMM couples the identification of metastable states with the estimation of transition probabilities. As such, accurate models require fewer states, aiding interpretability from a physical perspective.

The switch is not without tradeoffs. MSMs are backed by a significant body of theoretical work: the MSM is a direct discretization of an integral operator which formally controls the long timescale dynamics known as the transfer operator. This connection enables the quantification of approximation error in the MSM framework \cite{Prinz2011Markov}. No such theoretical guarantees yet exist for the $L_1$-regularized reversible HMM because the evolution of $Y_t$ is no longer unconditionally Markovian. However, because the HMM can be viewed as a generalized hidden MSM, there is reason to believe that analogues of MSM theoretical guarantees extend to the HMM framework.

While the $L_1$-regularized reversible hidden Markov model represents an improvement over previous methods for analyzing MD datasets, future work will likely confront a number of remaining challenges. For example, the current model does not learn the featurization and treats $\Delta$ as a hyperparameter. Bringing these two aspects of the model into the optimization framework would reduce the required amount of manual tuning. Adapting techniques from Bayesian nonparametrics, unsupervised feature learning and linear dynamical systems may facilitate the achievement of these goals.

Our results show that structured statistical analysis of massive protein datasets is now possible. We reduce complex dynamical systems with thousands of physical degrees of freedom to simple statistical models characterized by a small number of metastable states and transition rates. The HMM framework is a tool for turning raw molecular dynamics data into scientific knowledge about protein structure, dynamics and function. Our experiments on the ubiquitin and c-Src kinase proteins extract insight that may further the state of the art in cellular biology and rational drug design.

\section*{Acknowledgments}
We thank Diwakar Shukla for kindly providing the c-Src kinase MD trajectories. B.R. was supported by the Fannie and John Hertz Foundation. G.K and V.S.P acknowledge support from the Simbios NIH Center for Biomedical Computation (NIH U54 Roadmap GM072970). V.S.P. acknowledges NIH R01-GM62868 and NSF MCB-0954714.

\bibliography{l1_hmm.bib}

\newcommand{\noop}[1]{}
\begin{thebibliography}{35}
\providecommand{\natexlab}[1]{#1}
\providecommand{\url}[1]{\texttt{#1}}
\expandafter\ifx\csname urlstyle\endcsname\relax
  \providecommand{\doi}[1]{doi: #1}\else
  \providecommand{\doi}{doi: \begingroup \urlstyle{rm}\Url}\fi

\bibitem[Aleshin \& Finn(2010)Aleshin and Finn]{pmid20689754}
Aleshin, A. and Finn, R.~S.
\newblock {{S}{R}{C}: a century of science brought to the clinic}.
\newblock \emph{Neoplasia}, 12\penalty0 (8):\penalty0 599--607, 2010.

\bibitem[Baldi \& Pollastri(2003)Baldi and
  Pollastri]{baldi:2003:PDL:945365.945379}
Baldi, Pierre and Pollastri, Gianluca.
\newblock The principled design of large-scale recursive neural network
  architectures--{DAG}-{RNN}s and the protein structure prediction problem.
\newblock \emph{J. Mach, Learn. Res.}, 4:\penalty0 575--602, 2003.

\bibitem[Beauchamp et~al.(2011)Beauchamp, Bowman, Lane, Maibaum, Haque, and
  Pande]{doi:10.1021/ct200463m}
Beauchamp, Kyle~A., Bowman, Gregory~R., Lane, Thomas~J., Maibaum, Lutz, Haque,
  Imran~S., and Pande, Vijay~S.
\newblock {MSMB}uilder2: Modeling conformational dynamics on the picosecond to
  millisecond scale.
\newblock \emph{J. Chem. Theory Comput.}, 7\penalty0 (10):\penalty0 3412--3419,
  2011.

\bibitem[Beauchamp et~al.(2012)Beauchamp, Lin, Das, and
  Pande]{Beauchamp2012Protein}
Beauchamp, Kyle~A., Lin, Yu-Shan, Das, Rhiju, and Pande, Vijay~S.
\newblock Are protein force fields getting better? {A} systematic benchmark on
  524 diverse {NMR} measurements.
\newblock \emph{J. Chem. Theory Comput.}, 8\penalty0 (4):\penalty0 1409--1414,
  2012.

\bibitem[Bowman(2012)]{Bowman2012Improved}
Bowman, Gregory~R.
\newblock Improved coarse-graining of markov state models via explicit
  consideration of statistical uncertainty.
\newblock \emph{J. Chem. Phys.}, 137\penalty0 (13):\penalty0 134111, 2012.

\bibitem[Cho et~al.(2006)Cho, Levy, and Wolynes]{Cho17012006}
Cho, Samuel~S., Levy, Yaakov, and Wolynes, Peter~G.
\newblock P versus {Q}: Structural reaction coordinates capture protein folding
  on smooth landscapes.
\newblock \emph{Proc. Natl. Acad. Sci. U.S.A.}, 103\penalty0 (3):\penalty0
  586--591, 2006.

\bibitem[Chu et~al.(2006)Chu, Ghahramani, Podtelezhnikov, and
  Wild]{chu:2006:BSM:1137243.1137470}
Chu, Wei, Ghahramani, Zoubin, Podtelezhnikov, Alexei, and Wild, David~L.
\newblock Bayesian segmental models with multiple sequence alignment profiles
  for protein secondary structure and contact map prediction.
\newblock \emph{IEEE/ACM Trans. Comput. Biol. Bioinf.}, 3\penalty0
  (2):\penalty0 98--113, 2006.

\bibitem[Di~Lena et~al.(2012)Di~Lena, Baldi, and Nagata]{di2012deep}
Di~Lena, Pietro, Baldi, Pierre, and Nagata, Ken.
\newblock Deep spatio-temporal architectures and learning for protein structure
  prediction.
\newblock In \emph{Adv. Neural Inf. Process. Syst. 25}, NIPS '12, pp.\
  521--529, 2012.

\bibitem[Dill et~al.(1995)Dill, Bromberg, Yue, Chan, Ftebig, Yee, and
  Thomas]{dill1995principles}
Dill, Ken~A, Bromberg, Sarina, Yue, Kaizhi, Chan, Hue~Sun, Ftebig, Klaus~M,
  Yee, David~P, and Thomas, Paul~D.
\newblock Principles of protein folding -- a perspective from simple exact
  models.
\newblock \emph{Protein Sci.}, 4\penalty0 (4):\penalty0 561--602, 1995.

\bibitem[Fan \& Li(2001)Fan and Li]{fan2001variable}
Fan, Jianqing and Li, Runze.
\newblock Variable selection via nonconcave penalized likelihood and its oracle
  properties.
\newblock \emph{J. Am. Stat. Assoc.}, 96\penalty0 (456):\penalty0 1348--1360,
  2001.

\bibitem[Fang et~al.(2013)Fang, Grütter, and Rauh]{doi:10.1021/cb300663j}
Fang, Zhizhou, Grütter, Christian, and Rauh, Daniel.
\newblock Strategies for the selective regulation of kinases with allosteric
  modulators: Exploiting exclusive structural features.
\newblock \emph{ACS Chem. Biol.}, 8\penalty0 (1):\penalty0 58--70, 2013.

\bibitem[Guo et~al.(2010)Guo, Levina, Michailidis, and Zhu]{guo2010pairwise}
Guo, Jian, Levina, Elizaveta, Michailidis, George, and Zhu, Ji.
\newblock Pairwise variable selection for high-dimensional model-based
  clustering.
\newblock \emph{Biometrics}, 66\penalty0 (3):\penalty0 793--804, 2010.

\bibitem[Hershko \& Ciechanover(1998)Hershko and
  Ciechanover]{hershko1998ubiquitin}
Hershko, Avram and Ciechanover, Aaron.
\newblock The ubiquitin system.
\newblock \emph{Annu. Rev. Biochem.}, 67\penalty0 (1):\penalty0 425--479, 1998.

\bibitem[Higham(2001)]{higham2001algorithmic}
Higham, Desmond~J.
\newblock An algorithmic introduction to numerical simulation of stochastic
  differential equations.
\newblock \emph{SIAM Rev.}, 43\penalty0 (3):\penalty0 525--546, 2001.

\bibitem[Humphrey et~al.(1996)Humphrey, Dalke, and Schulten]{Humphrey199633}
Humphrey, William, Dalke, Andrew, and Schulten, Klaus.
\newblock {VMD}: Visual molecular dynamics.
\newblock \emph{J. Mol. Graphics}, 14\penalty0 (1):\penalty0 33 -- 38, 1996.

\bibitem[Karplus \& Kuriyan(2005)Karplus and Kuriyan]{Karplus10052005}
Karplus, M. and Kuriyan, J.
\newblock Molecular dynamics and protein function.
\newblock \emph{Proc. Natl. Acad. Sci. U.S.A.}, 102\penalty0 (19):\penalty0
  6679--6685, 2005.

\bibitem[Kohlhoff et~al.(2014)Kohlhoff, Shukla, Lawrenz, Bowman, Konerding,
  Belov, Altman, and Pande]{kohlhoff2014cloud}
Kohlhoff, Kai~J, Shukla, Diwakar, Lawrenz, Morgan, Bowman, Gregory~R,
  Konerding, David~E, Belov, Dan, Altman, Russ~B, and Pande, Vijay~S.
\newblock Cloud-based simulations on {Google Exacycle} reveal ligand modulation
  of gpcr activation pathways.
\newblock \emph{Nat. Chem.}, 6\penalty0 (1):\penalty0 15--21, 2014.

\bibitem[Komander et~al.(2009)Komander, Clague, and
  Urb{\'e}]{komander2009breaking}
Komander, David, Clague, Michael~J, and Urb{\'e}, Sylvie.
\newblock Breaking the chains: structure and function of the deubiquitinases.
\newblock \emph{Nat. Rev. Mol. Cell Biol.}, 10\penalty0 (8):\penalty0 550--563,
  2009.

\bibitem[Lane et~al.(2013)Lane, Shukla, Beauchamp, and
  Pande]{lane2012milliseconds}
Lane, Thomas~J, Shukla, Diwakar, Beauchamp, Kyle~A, and Pande, Vijay~S.
\newblock To milliseconds and beyond: challenges in the simulation of protein
  folding.
\newblock \emph{Curr. Opin. Struct. Biol.}, 23\penalty0 (1):\penalty0 58 -- 65,
  2013.

\bibitem[Maity et~al.(2005)Maity, Maity, Krishna, Mayne, and
  Englander]{Maity29032005}
Maity, Haripada, Maity, Mita, Krishna, Mallela M.~G., Mayne, Leland, and
  Englander, S.~Walter.
\newblock Protein folding: The stepwise assembly of foldon units.
\newblock \emph{Proc. Natl. Acad. Sci. U.S.A.}, 102\penalty0 (13):\penalty0
  4741--4746, 2005.

\bibitem[Nocedal \& Wright(2006)Nocedal and Wright]{nocedal2006numerical}
Nocedal, J. and Wright, S.
\newblock \emph{Numerical Optimization}.
\newblock Springer Series in Operations Research and Financial Engineering.
  Springer, New York, 2nd edition, 2006.

\bibitem[No{\'e} et~al.(2013)No{\'e}, Wu, Prinz, and
  Plattner]{:/content/aip/journal/jcp/139/18/10.1063/1.4828816}
No{\'e}, Frank, Wu, Hao, Prinz, Jan-Hendrik, and Plattner, Nuria.
\newblock Projected and hidden markov models for calculating kinetics and
  metastable states of complex molecules.
\newblock \emph{J. Chem. Phys.}, 139\penalty0 (18):\penalty0 184114, 2013.

\bibitem[Prinz et~al.(2011)Prinz, Wu, Sarich, Keller, Senne, Held, Chodera,
  Sch{\"u}tte, and No{\'e}]{Prinz2011Markov}
Prinz, Jan-Hendrik, Wu, Hao, Sarich, Marco, Keller, Bettina, Senne, Martin,
  Held, Martin, Chodera, John~D., Sch{\"u}tte, Christof, and No{\'e}, Frank.
\newblock Markov models of molecular kinetics: Generation and validation.
\newblock \emph{J. Chem. Phys.}, 134\penalty0 (17):\penalty0 174105, 2011.

\bibitem[Rabiner(1989)]{rabiner1986introduction}
Rabiner, Lawrence~R.
\newblock A tutorial on hidden markov models and selected applications in
  speech recognition.
\newblock \emph{Proc. IEEE}, 77\penalty0 (2):\penalty0 257--286, 1989.

\bibitem[Sadiq et~al.(2012)Sadiq, No\'{e}, and De~Fabritiis]{Sadiq26112012}
Sadiq, S.~Kashif, No\'{e}, Frank, and De~Fabritiis, Gianni.
\newblock Kinetic characterization of the critical step in {HIV-1} protease
  maturation.
\newblock \emph{Proc. Natl. Acad. Sci. U.S.A.}, 109\penalty0 (50):\penalty0
  20449--20454, 2012.

\bibitem[Shaw(2007)]{Shaw:2007:ASM:1250662.1250664}
Shaw, David E. et~al.
\newblock Anton, a special-purpose machine for molecular dynamics simulation.
\newblock In \emph{ACM Comp. Ar. 34}, ISCA '07, pp.\  1--12, 2007.

\bibitem[Shirts \& Pande(2000)Shirts and Pande]{Shirts08122000}
Shirts, Michael and Pande, Vijay~S.
\newblock Screen savers of the world unite!
\newblock \emph{Science}, 290\penalty0 (5498):\penalty0 1903--1904, 2000.

\bibitem[Shukla et~al.(2014)Shukla, Meng, Roux, and
  Pande]{Shukla2014Activation}
Shukla, Diwakar, Meng, Yilin, Roux, Beno\^{i}t, and Pande, Vijay~S.
\newblock {A}ctivation pathway of {S}rc kinase reveals intermediate states as
  novel targets for drug design.
\newblock \emph{Nat. Commun.}, 5, 2014.

\bibitem[Sontag et~al.(2012)Sontag, Meltzer, Globerson, Jaakkola, and
  Weiss]{sontag2012tightening}
Sontag, David, Meltzer, Talya, Globerson, Amir, Jaakkola, Tommi~S, and Weiss,
  Yair.
\newblock Tightening {LP} relaxations for {MAP} using message passing.
\newblock \emph{arXiv preprint arXiv:1206.3288}, 2012.

\bibitem[Tibshirani et~al.(2005)Tibshirani, Saunders, Rosset, Zhu, and
  Knight]{tibshirani2005sparsity}
Tibshirani, Robert, Saunders, Michael, Rosset, Saharon, Zhu, Ji, and Knight,
  Keith.
\newblock Sparsity and smoothness via the fused lasso.
\newblock \emph{J. R. Statistic. Soc. B}, 67\penalty0 (1):\penalty0 91--108,
  2005.

\bibitem[Voelz et~al.(2012)Voelz, Jäger, Yao, Chen, Zhu, Waldauer, Bowman,
  Friedrichs, Bakajin, Lapidus, Weiss, and Pande]{doi:10.1021/ja302528z}
Voelz, Vincent~A., Jäger, Marcus, Yao, Shuhuai, Chen, Yujie, Zhu, Li,
  Waldauer, Steven~A., Bowman, Gregory~R., Friedrichs, Mark, Bakajin, Olgica,
  Lapidus, Lisa~J., Weiss, Shimon, and Pande, Vijay~S.
\newblock Slow unfolded-state structuring in {Acyl-CoA} binding protein folding
  revealed by simulation and experiment.
\newblock \emph{J. Am. Chem. Soc.}, 134\penalty0 (30):\penalty0 12565--12577,
  2012.

\bibitem[Weber \& Pande(2012)Weber and Pande]{weber2012protein}
Weber, Jeffrey~K and Pande, Vijay~S.
\newblock Protein folding is mechanistically robust.
\newblock \emph{Biophys. J.}, 102\penalty0 (4):\penalty0 859--867, 2012.

\bibitem[Wong \& McCammon(2003)Wong and
  McCammon]{doi:10.1146/annurev.pharmtox.43.100901.140216}
Wong, Chung~F. and McCammon, J.~Andrew.
\newblock Protein flexibility and computer-aided drug design.
\newblock \emph{Annu. Rev. Pharmacol. Toxicol.}, 43\penalty0 (1):\penalty0
  31--45, 2003.

\bibitem[Zhang et~al.(2012)Zhang, Zhou, Rouge, Phillips, Lam, Liu, Sandoval,
  Helgason, Murray, Wertz, et~al.]{zhang2012conformational}
Zhang, Yingnan, Zhou, Lijuan, Rouge, Lionel, Phillips, Aaron~H, Lam, Cynthia,
  Liu, Peter, Sandoval, Wendy, Helgason, Elizabeth, Murray, Jeremy~M, Wertz,
  Ingrid~E, et~al.
\newblock Conformational stabilization of ubiquitin yields potent and selective
  inhibitors of {USP}7.
\newblock \emph{Nat. Chem. Biol.}, 9\penalty0 (1):\penalty0 51--58, 2012.

\bibitem[Zhuang et~al.(2011)Zhuang, Cui, Silva, and
  Huang]{doi:10.1021/jp109592b}
Zhuang, Wei, Cui, Raymond~Z., Silva, Daniel-Adriano, and Huang, Xuhui.
\newblock Simulating the {T}-jump-triggered unfolding dynamics of trpzip2
  peptide and its time-resolved {IR} and two-dimensional {IR} signals using the
  markov state model approach.
\newblock \emph{J. Phys. Chem. B}, 115\penalty0 (18):\penalty0 5415--5424,
  2011.

\end{thebibliography}
\bibliographystyle{icml2014}
\end{document}